\documentclass{article}

\pdfoutput=1

\usepackage[margin=.8in]{geometry}

\usepackage{authblk}

\usepackage{cite}
\usepackage{amsmath,amssymb,amsfonts}
\usepackage{graphicx}
\usepackage{textcomp}
\usepackage{xcolor}
\usepackage{booktabs}
\def\BibTeX{{\rm B\kern-.05em{\sc i\kern-.025em b}\kern-.08em
    T\kern-.1667em\lower.7ex\hbox{E}\kern-.125emX}}

\usepackage{epsfig}
\usepackage{amsmath}
\usepackage{balance}  
\usepackage{cleveref}
\usepackage{babel}
\addto\languageXYZ{\renewcommand\proofname{ABC}}

\usepackage{lipsum}
\usepackage{algorithm}

\usepackage{amsthm}
\usepackage{algpseudocode}
\usepackage{bm}
\usepackage{comment}
\usepackage{array}
\usepackage{float}
\floatstyle{plaintop}
\restylefloat{table}
\ifdefined\VLDB
\newcommand{\subparagraph}{}
\fi

\usepackage{amsmath}
\usepackage{algorithm}
\usepackage{subfig}
\usepackage{blindtext}
\usepackage{multirow}
\usepackage{xcolor}
\usepackage{url}
\usepackage{balance}
\usepackage{nopageno}
\usepackage[tableposition=top]{caption}
\usepackage[labelfont=bf]{caption}
\usepackage{graphicx}
\usepackage{tikz}
\usepackage{tcolorbox}
\usepackage{tcolorbox}
\usepackage{tikz}
\usetikzlibrary{shapes.geometric}

\begin{document}

\date{}



\title{Learning-Augmented Frequency Estimation in Sliding Windows}

\author[1]{Rana Shahout}

\author[2]{Ibrahim Sabek}

\author[1]{Michael Mitzenmacher}

\affil[1]{Harvard University, USA}
\affil[2]{University of Southern California, USA}

\newif\ifcomm
\commtrue
\ifcomm
\else
\commfalse
\fi
\ifcomm
	\newcommand{\mycomm}[3]{{\footnotesize{{\color{#2} \textbf{[#1: #3]}}}}}
	\newcommand{\CRdel}[1]{\textcolor{red}{\sout{#1}}}
\else
    \newcommand{\mycomm}[3]{}
    \newcommand{\CRdel}[1]{}
\fi
\newcommand{\rana}[1]{\mycomm{Rana}{red}{#1}} 
\def\delequal{\mathrel{\ensurestackMath{\stackon[1pt]{=}{\scriptscriptstyle\Delta}}}}

\newtheorem{theorem}{Theorem}
\newtheorem{definition}{Definition}
\newtheorem{corollary}{Corollary}
\newtheorem{lemma}{Lemma}
\newtheorem{observation}[theorem]{Observation}
\newenvironment{sproof}{%
	\renewcommand{\proofname}{Proof Sketch}\proof}{\endproof}

\newcommand{\HHLproblem}{$(\theta,\epsilon,\delta)$-HH-\revision{quantiles}}
\newcommand{\HHLproblemParams}[3]{$(#1,#2,#3)$-HH-\revision{quantiles}}

\newcommand{\WHHLproblem}{$(\theta,\epsilon,\delta)$-$\mathit{Weighted}$-HH-quantiles}
\newcommand{\WHHWLproblem}{$(\theta,\epsilon,\delta)$-$\mathit{Weighted}^2$-HH-quantiles}

\ifdefined\EXTENDED
\newcommand{\matrixCellWidth}{5.8cm}
\else
\newcommand{\matrixCellWidth}{5.5cm}
\fi
\newcommand{\smatrixCellWidth}{4.5cm}

\newcommand{\tableDis}{d}
\renewcommand{\mid}{\ensuremath{:}}
\newcommand{\curFrameFreq}{\ensuremath{\mathit{cFreq}}}
\newcommand{\eps}{\epsilon}
\newcommand{\brackets}[1]{\ensuremath{\left[#1\right]}}
\newcommand{\BS}{\ensuremath{\mathit{BlockSize}}}
\newcommand{\Tables}{\ensuremath{\mathit{Tables}}}
\newcommand{\incTables}{\ensuremath{\mathit{incTables}}}
\newcommand{\incTable}{\ensuremath{\mathit{incTable}}}
\newcommand{\ghostTables}{\ensuremath{\mathit{ghostTables}}}
\newcommand{\OFFSET}{\ensuremath{\mathit{offset}}}
\newcommand{\blockSize}{\ensuremath{\mathit{s}}}
\newcommand{\IFText}{IntervalFrequency}
\newcommand{\WText}{WindowFrequency}
\newcommand{\QWText}{WFrequency}
\newcommand{\FText}{Frequency}

\newcommand{\IFTextD}{Interval-Frequency}
\newcommand{\TIFText}{TimeIntervalFrequency}
\newcommand{\IHHText}{IntervalHeavyHitters}
\newcommand{\SIText}{Interval}
\newcommand{\SIproblemParams}{n}
\newcommand{\SIProblem}{${\SIproblemParams}${-}{\SIText}}
\newcommand{\frameOffset}{\ensuremath{\mathit{fo}}}
\newcommand{\SSInstance}{\ensuremath{\mathit{SS}}}
\newcommand{\SIQText}{\SIText{}Query}
\newcommand{\IFQText}{\IFText{}Query}
\newcommand{\TIFQText}{\TIFText{}Query}
\newcommand{\IHHQText}{\IHHText{}Query}
\newcommand{\problemParams}{(W,\epsilon)}

\newcommand{\timeProblemParams}{(\Tau,R,\epsilon)}
\newcommand{\IFProblem}{${\problemParams}${-}{\IFText}}
\newcommand{\WindowProblem}{${\problemParams}${-}{\WText}}
\newcommand{\QWindowProblem}{${\problemParams}${-}{\QWText}}
\newcommand{\FrequencyProblem}{${\problemParams}${-}{\QWText}}

\newcommand{\IFProblemD}{${\problemParams}${-}{\IFTextD}}
\newcommand{\ParameterizedIFProblem}[2]{${(#1,#2)}${-}{\IFText}}
\newcommand{\TIFProblem}{${\timeProblemParams}${-}{\TIFText}}
\newcommand{\IHHProblem}{${\problemParams}$\textbf{-}{\IHHText{}}}
\newcommand{\BIFProblem}{$\mathbf{\problemParams}${-}{\IFText}}
\newcommand{\BIHHProblem}{$\mathbf{\problemParams}$\textbf{-}{\IHHText{}}}
\newcommand{\IFQGuarantee}[1][]{$$\xIntervalFrequency \le \xIntervalFrequencyEstimator \le \xIntervalFrequency + \epsError.#1$$}
\newcommand{\WGuarantee}[1][]{$$\iFrequency \le \iFrequencyEstimator \le \iFrequency + \epsError.#1$$}

\newcommand{\intervalHHThreshold}{\ensuremath{\theta\cdotpa{j-i}}}
\newcommand{\IHHDef}{\ensuremath{\set{x\in\mathcal U\mid \xIntervalFrequency\ge \intervalHHThreshold}}}
\newcommand{\xWindowFrequency}[1][w]{\ensuremath{f^{#1}_i}}
\newcommand{\xTimeWindowFrequency}[1][t]{\ensuremath{h^{#1}_x}}
\newcommand{\xWindowFrequencyEstimator}{\ensuremath{\widehat{\xWindowFrequency}}}
\newcommand{\xIntervalFrequency}{\ensuremath{f^{k,j}_i}}
\newcommand{\iFrequency}{\ensuremath{f_i}}
\newcommand{\xTimeIntervalFrequency}{\ensuremath{h^{i,j}_x}}
\newcommand{\xTimeIntervalFrequencyEstimator}{\ensuremath{\widehat{\xTimeIntervalFrequency}}}
\newcommand{\xSetIntervalFrequency}{\ensuremath{g^{i, j}_x}}
\newcommand{\xSetWindowFrequency}[1][n]{\ensuremath{g^{#1}_x}}
\newcommand{\xSetWindowFrequencyEstimator}{\ensuremath{\widehat{\xSetWindowFrequency}}}
\newcommand{\xSetIntervalFrequencyEstimator}{\ensuremath{\widehat{\xSetIntervalFrequency}}}
\newcommand{\xParameterizedIntervalFrequency}[2]{\ensuremath{f^{#1,#2}_x}}
\newcommand{\xIntervalFrequencyEstimator}{\ensuremath{\widehat{\xIntervalFrequency}}}
\newcommand{\iFrequencyEstimator}{\ensuremath{\widehat{\iFrequency}}}
\newcommand{\IntervalHH}[1][\theta]{\ensuremath{HH_{#1}^{i,j}}}
\newcommand{\IntervalHHEstimator}{\ensuremath{\widehat{\IntervalHH}}}
\newcommand{\epsError}[1][W]{\ensuremath{#1 \epsilon}}
\newcommand{\epsNError}[1][N]{\ensuremath{#1 \epsilon}}

\newcommand{\IFQ}{{\sc {\IFQText}}}
\newcommand{\SIQ}{{\sc {\SIQText}}}
\newcommand{\IHHQ}{{\sc {\IHHQText}}}
\newcommand{\BIFQ}{{\sc \textbf{\IFQText}}}
\newcommand{\BTIFQ}{{\sc \textbf{\TIFQText}}}
\newcommand{\BSIQ}{{\sc \textbf{\SIQText}}}
\newcommand{\BIHHQ}{{\sc \textbf{\IHHQText}}}
\newcommand{\set}[1]{\left\{#1\right\}}
\newcommand{\ceil}[1]{ \left\lceil{#1}\right\rceil}
\newcommand{\floor}[1]{ \left\lfloor{#1}\right\rfloor}
\newcommand{\angles}[1]{ \left\langle{#1}\right\rangle}
\newcommand{\logp}[1]{\log\parentheses{#1}}
\newcommand{\clog}[1]{ \ceil{\log{#1}}}
\newcommand{\clogp}[1]{ \ceil{\logp{#1}} }
\newcommand{\flog}[1]{ \floor{\log{#1}}}
\newcommand{\parentheses}[1]{ \left({#1}\right)}
\newcommand{\abs}[1]{ \left|{#1}\right|}

\newcommand{\cdotpa}[1]{\cdot\parentheses{#1}}
\newcommand{\inc}[1]{$#1 \gets #1 + 1$}
\newcommand{\dec}[1]{$#1 \gets #1 - 1$}
\newcommand{\range}[2][0]{#1,1,\ldots,#2}
\newcommand{\frange}[1]{\set{\range{#1}}}
\newcommand{\xrange}[1]{\frange{#1-1}}
\newcommand{\oneOverE}{ \ensuremath{\eps^{-1}} }
\newcommand{\oneOverG}{ \frac{1}{\gamma} }
\newcommand{\oneOverT}{ \frac{1}{\tau} }
\newcommand{\smallMultError}{(1+o(1))}
\newcommand{\lowerbound}{\max \set{\log W ,\frac{1}{2\epsilon+W^{-1}}}}
\newcommand{\smallEpsLowerbound}{\window\logp{\frac{1}{\weps}}}
\newcommand{\smallEpsMemoryTheta}{$\Theta\parentheses{\smallEpsMemoryConsumption}$}
\newcommand{\smallEpsMemoryConsumption}{W\cdot\logp{\frac{1}{\weps}}}

\newcommand{\largeEpsRestriction}{For any \largeEps{},}
\newcommand{\largeEps}{$\eps^{-1} \le 2W\left(1-\frac{1}{\logw}\right)$}
\newcommand{\smallEpsRestriction}{For any \smallEps{},}
\newcommand{\smallEps}{$\eps^{-1}>2W\left(1-\frac{1}{\logw}\right)=2\window(1-o(1))$}
\newcommand{\bc}{{\sc Basic-Counting}}
\newcommand{\bs}{{\sc Basic-Summing}}
\newcommand{\windowcounting}{ {\sc $(W,\epsilon)$-Window-Counting}}

\newcommand{\window}{W}
\newcommand{\logw}{\log \window}
\newcommand{\flogw}{\floor{\log \window}}
\newcommand{\weps}{\window\epsilon}
\newcommand{\wt}{\window\tau}
\newcommand{\logweps}{\logp{\weps}}
\newcommand{\logwt}{\logp{\wt}}
\newcommand{\bitarray}{b}
\newcommand{\currentBlockIndex}{i}
\newcommand{\currentBlock}{\bitarray_{\currentBlockIndex}}
\newcommand{\remainder}{y}
\newcommand{\numBlocks}{n}
\newcommand{\sumOfBits}{B}
\newcommand{\iblockSize}{\frac{\numBlocks}{\window}}
\newcommand{\threshold}{\blockSize}
\newcommand{\halfBlock}{\frac{\window}{2\numBlocks}}
\newcommand{\blockOffset}{m}
\newcommand{\inputVariable}{x}

\newcommand{\bcTableColumnWidth}{1.5cm}
\newcommand{\bsTableColumnWidth}{1.7cm}
\newcommand{\bsExtendedTableColumnWidth}{3cm}
\newcommand{\bcExtendedTableColumnWidth}{2.8cm}
\newcommand{\bcNarrowTableColumnWidth}{1.5cm}
\newcommand{\bsNarrowTableColumnWidth}{1.5cm}
\newcommand{\bsWorstCaseTableColumnWidth}{2cm}

\newcommand{\bsrange}{ R }
\newcommand{\bsReminderPercisionParameter}{ \gamma }
\newcommand{\bsest}{ \widehat{\bssum}}
\newcommand{\bssum}{ S^W }
\newcommand{\bsFracInput}{ \inputVariable' }
\newcommand{\bserror}{ \bsrange\window\epsilon }
\newcommand{\bsfractionbits}{ \frac{\bsReminderPercisionParameter}{\epsilon} }
\newcommand{\bsReminderFractionBits}{ \upsilon}
\newcommand{\bsAnalysisTarget}{ \bssum}
\newcommand{\bsAnalysisEstimator}{ \widehat{\bsAnalysisTarget}}
\newcommand{\bsAnalysisError}{ \bsAnalysisEstimator - \bsAnalysisTarget}
\newcommand{\bsRoundingError}{ \xi}


\newcommand{\neps}{\ensuremath{\winSize\eps}}
\newcommand{\Neps}{\ensuremath{\maxWinSize\eps}}
\newcommand{\logn}{\ensuremath{\log\winSize}}
\newcommand{\logN}{\ensuremath{\log\maxWinSize}}
\newcommand{\logneps}{\ensuremath{\logp\neps}}
\newcommand{\logNeps}{\ensuremath{\logp\Neps}}
\newcommand{\oneOverEps}{\ensuremath{\frac{1}{\eps}}}
\newcommand{\winSize}{\ensuremath{n}}
\newcommand{\maxWinSize}{\ensuremath{N}}
\newcommand{\curTime}{\ensuremath{t}}
\newcommand\Tau{\mathrm{T}}
\newcommand{\offset}{\ensuremath{\mathit{offset}}}
\newcommand{\roundedOOE}{k}
\newcommand{\numLargeBlocks}{\frac{\roundedOOE}{4}}
\newcommand{\numSmallBlocks}{\frac{\roundedOOE}{2}}

\newcommand{\remove}{{\sc Remove()}}
\newcommand{\merge}[1]{{\sc Merge(#1)}}
\newcommand{\counting}{{\sc Counting}}
\newcommand{\summing}{{\sc Summing}}
\newcommand{\freq}{{\sc Frequency Estimation}}

\newcommand{\RSS}{Rounded Space Saving}
\newcommand{\rss}{RSS}
\newcommand{\slack}{1+\gamma}
\newcommand{\nrCounters}{\ceil{\frac{\slack}{\epsilon}}}
\newcommand{\nrCountersLetter}{\mathfrak{C}}
\newcommand{\step}{\ensuremath{\floor{\frac{M\cdot\gamma}{2} + 1}}}
\newcommand{\stepLetter}{\mathpzc{s}}
\newcommand{\frqEst}{$(\epsilon,M)$-{\sc Volume Estimation}}
\newcommand{\heavyHitters}{$(\theta,\epsilon,M)$-{\sc Weighted Heavy Hitters}}
\newcommand{\query}{{\sc Query$(x)$}}
\newcommand{\winQuery}{{\sc WinQuery$(x)$}}
\newcommand{\windowQueryText}{{\sc WinQuery}}
\newcommand{\windowQuery}{\windowQueryText$(x,w)$}
\newcommand{\interWinQuery}{{\sc InterWinQuery$(x, i, j)$}}

\newcommand{\SSS}{CSS}
\newcommand{\SpaceS}{Space-Saving}
\newcommand{\CSS}{Compact \SpaceS{}}
\newcommand{\WCSS}{Window \CSS{}}
\newcommand{\SSSInstance}{rss}
\newcommand{\queueOfOverflows}{b}
\newcommand{\sumOfOverflows}{B}
\newcommand{\skiplist}{A}
\newcommand{\idtoidx}{IDToIndex}
\newcommand{\IDArray}{O}
\newcommand{\overflowidx}{idx}
\newcommand{\overflowIndicator}{u}

\newcommand{\streamcounting}{$\epsilon$\textsc{-Counting}}
\newcommand{\paremeterizedStreamcounting}[1]{#1\textsc{-Counting}}
\newcommand{\probabilisticStreamcounting}{$(\epsilon, \delta)$\textsc{-Counting}}
\newcommand{\probabilisticWindowcounting}{$(\window,\epsilon, \delta)$\textsc{-Window Counting}}
\newcommand{\reverseMapping}{$\mathbf R$}
\renewcommand{\gamma}{\phi}

\renewcommand{\epsilon}{\varepsilon}
\newcommand{\revision}[1]{\textcolor{black}{#1}} 

\newcommand{\lss}{LSS\xspace}
\newcommand{\lssFULL}{Learned Space Saving}

\newcommand{\lssplus}{LSS+\xspace}

\newcommand{\lsscbf}{LSS-CBF\xspace}
\newcommand{\lssasn}{LSS-AssignedEntries}
\newcommand{\lssdual}{LSS-Dual\xspace}

\newcommand{\alg}{LWCSS\xspace}
\newcommand{\algcbf}{LSS-CBF\xspace}

\newcommand{\optalg}{LWCSS+\xspace}

\maketitle

\begin{abstract}
We show how to utilize machine learning approaches to improve sliding window algorithms for approximate frequency estimation problems, under the ``algorithms with predictions'' framework.
In this dynamic environment, previous learning-augmented algorithms are less effective, since properties in sliding window resolution can differ significantly from the properties of the entire stream.
Our focus is on the benefits of predicting and filtering out items with large next arrival times -- that is, there is a large gap until their next appearance -- from the stream, which we show improves the memory-accuracy tradeoffs significantly.
We provide theorems that provide insight into how and by how much our technique can improve the sliding window algorithm, as well as experimental results using real-world data sets.
Our work demonstrates that predictors can be useful in the challenging sliding window setting.
\end{abstract}

\section{Introduction}
\label{sec:intro}

Stream processing plays a crucial role in various applications, such as network monitoring, intrusion detection systems, and sensor networks. Dealing with large and rapidly incoming data streams presents challenges in providing accurate responses to queries due to high computational and memory costs. Furthermore, these applications require time-efficient algorithms to cope with high-speed data streams. To that end, stream processing algorithms often build compact approximate sketches (synopses) of the input streams.

Estimating the frequency of certain items in a data stream is a fundamental step in data analysis. 
Several algorithms, such as Count-Min Sketch~\cite{cormode2005improved}, Count Sketch~\cite{charikar2002finding} have been proposed to estimate item frequencies in the data streams.

As time passes, newer data often becomes more relevant than older data, necessitating an aging mechanism for the sketches. 
In financial analysis, for instance, analysts prioritize current market trends, whereas in intrusion detection systems, recent intrusion attempts are of primary concern. In both cases, outdated information loses significance over time. Retaining old data not only consumes valuable memory resources but also introduces noise, complicating the analysis of recent, relevant data. 
Many applications realize this by tracking the stream's items over a sliding window. The sliding window model~\cite{datar2002maintaining} considers only a window of the most recent items in the stream, while older ones do not affect the quantity we wish to estimate. Indeed, standard approaches to the problem of maintaining different types of sliding window statistics have been extensively studied~\cite{datar2002maintaining, arasu2004approximate, ben2016heavy, papapetrou2015sketching, lee2006simpler}.


Recently, machine learning has been combined with traditional algorithms, yielding the paradigm of learning-augmented algorithms, also known as algorithms with predictions.
This approach aims to improve algorithms by leveraging advice from machine learning models in the form of predictions.
In the context of streaming algorithms, Hsu et al~\cite{hsu2019learning} introduced learning-based frequency estimation, where machine learning is utilized to predict the most frequent items, known as 'heavy hitters,' with the goal of reducing estimation errors. 
They utilized known hashing-based algorithms such as Count-Sketch and Count-Min Sketch, which approximately count item frequencies by hashing items to buckets. The learning process ensured that predicted heavy hitter items were not placed in the sketch but assigned dedicated counters, facilitating accurate counting.
In addition, \cite{shahout2024learning} described a learned-based algorithm for identifying heavy hitters, top k, and flow frequency estimation. They focus on the Space Saving algorithm and exclude predicted low-frequency items from updating the Space Saving, while ensuring that predicted heavy-hitter items remain in the structure. In this way, the Space Saving accuracy is improved.


However, these approaches do not directly translate to the sliding window model, where different items can become heavy hitters or have low frequency locally within a window as time passes. In other words, an item that is a heavy hitter for the entire stream may not be a heavy hitter in any particular data window.
Moreover, to the best of our knowledge, predictions have not been applied to the problem of approximate counting in the sliding window setting. In general, sliding window variations of approximate counting and other problems are generally more challenging than variations without sliding windows, so perhaps this is unsurprising.

In this work, we aim to demonstrate the applicability of predictions to sliding window algorithms for frequency estimation problems, focusing on how natural predictions of the gap between item arrivals can lead to easy-to-implement improvements in this setting.

We focus on Window Compact Space-Saving (WCSS)~\cite{ben2016heavy} and present the learned-based version, \alg{}.
Our primary idea is to exclude items that appear only once \textbf{in the sliding window} from being included in the data structure used for tracking item frequencies, inspired by LSS~\cite{shahout2024learning}.

However, LSS employs predictors that relate to the entire stream, which cannot be directly applied to a sliding window. An item with low frequency in the entire stream might not be low-frequency within a particular window, potentially leading to unbounded estimation errors.

Our key insight is the potential to exclude items that appear once by employing an effective predictor for each item's next arrival time in the stream. With a perfectly accurate predictor, we could identify single-occurrence items within the sliding window by determining whether they have not appeared previously in the stream, and
have a predicted next appearance time that exceeds the current window size. Predictors of next arrivals have been employed in for example caching with machine learning advice~\cite{lykouris2021competitive}.

We treat the next arrival predictor as a black box and do not delve into its internal functioning; our approach can therefore be utilized with any suitable learning scheme that produces a predictor. We analyze both theoretically and empirically the performance gains offered by predictors, as well as the robustness of our approach to predictor~errors.

We also explore potential future directions, such as developing a predictor capable of learning item frequencies within a given time frame and periodically retraining itself using transfer learning techniques to adapt to shifts in the data distribution across frames. Furthermore, we discuss the prospect of extending our approach to handle other types of queries in the sliding window setting.

\section{Problem Definition}
\label{sec:prelim}

Given a \emph{universe} $\mathcal{U}$, a \emph{stream} $\mathcal{S} = u_1, u_2, \ldots \in \mathcal{U}^N$ is a sequence of arrivals from the universe. (We assume the stream is finite here.)
We denote by $f_i^{W}$ the frequency of item $i$ over the last $W$ arrivals.

We seek algorithms that support the following operations:
\begin{itemize}
	\item {\sc \textbf{ADD}$\bm{(i)}$}: given an element $i\in\mathcal U$, append $i$ to $\mathcal S$.
	\item {\sc \textbf{Query}$\bm{(i)}$}: return an estimate $\widehat{f_i^{W}}$ of $f_i^{W}$
\end{itemize}

\begin{definition}
\label{definition:wfrequency_error_guarantee}
An algorithm solves \QWindowProblem{} for if given any Query$(i)$ it returns $\widehat{f_i^{W}}$ satisfying
$$f_i^{W} \le \widehat{f_i^{W}} \le f_i^{W}  + \epsError.$$

\end{definition}

\section{Sliding Window Algorithms}
\label{sec:w_algs}

Sliding window algorithms have two conceptual steps: removing outdated data and recording incoming data.
These algorithms aim to avoid storing the entire window sequence, as the window size is typically large. 
Due to such memory limitations, approximation algorithms are desirable and often employed.  
Such algorithms typically do not maintain individual counters for each item, but instead employ an approximate counting algorithm to monitor item~frequencies.

In this work, we choose Window Compact Space Saving (WCSS)~\cite{ben2016heavy} as our underlying algorithm\footnote{Our approach could be applied to other sliding window streaming algorithms that follow the same setup of dividing the stream into frames (e.g.~\cite{basat2018stream, gou2020sliding})}.
WCSS solves \QWindowProblem{} when the queries are frequency queries.
It divides the stream into frames of size $W$ and each frame is divided into $k$ equal-sized blocks as illustrated in Figure~\ref{fig:window_setup}.
The query window also has size $W$ and, importantly, it overlaps with no more than 2 frames at any given time.

WCSS counts how many times each item arrived during the last frame.
When the counter exceeds a multiple of the block size ($\frac{W}{k}$, $\frac{2W}{k}$, etc), WCSS identifies it as an overflow.
The algorithm keeps track of the item ID associated with each overflow,
and selectively keeps only overflowed items for past blocks. To identify these overflowed items, WCSS uses an instance of the Space Saving algorithm (SS)~\cite{metwally2005efficient} which is re-initialized as empty at the beginning of each frame.
In addition, WCSS employs specific data structures to keep track of the overflowed item IDs, denoted by \emph{overflowsRecord}.

In WCSS, \emph{overflowsRecord} contains a collection of up to $k + 1$ queues.
Each queue corresponds to a block that overlaps with the current window. Within each queue, WCSS stores the IDs of items that have overflowed in the corresponding block. Whenever an item overflows, WCSS appends its ID to the current block’s queue. When a block ends, WCSS removes the oldest queue from the collection.
In order to estimate the frequency of an item within a given window, WCSS counts the overflow occurrences of the item and multiplies the result by $\frac{W}{k}$.
Similar setups are found in works that support other queries over sliding windows, such as in~\cite{basat2018stream}. In this work, \emph{overflowsRecord} contains a hierarchical tree structure consisting of frequency tables. These tables store overflowed item IDs along with their respective frequencies within block intervals.

\begin{figure}
    \centering
    \includegraphics[width = 6.5cm]{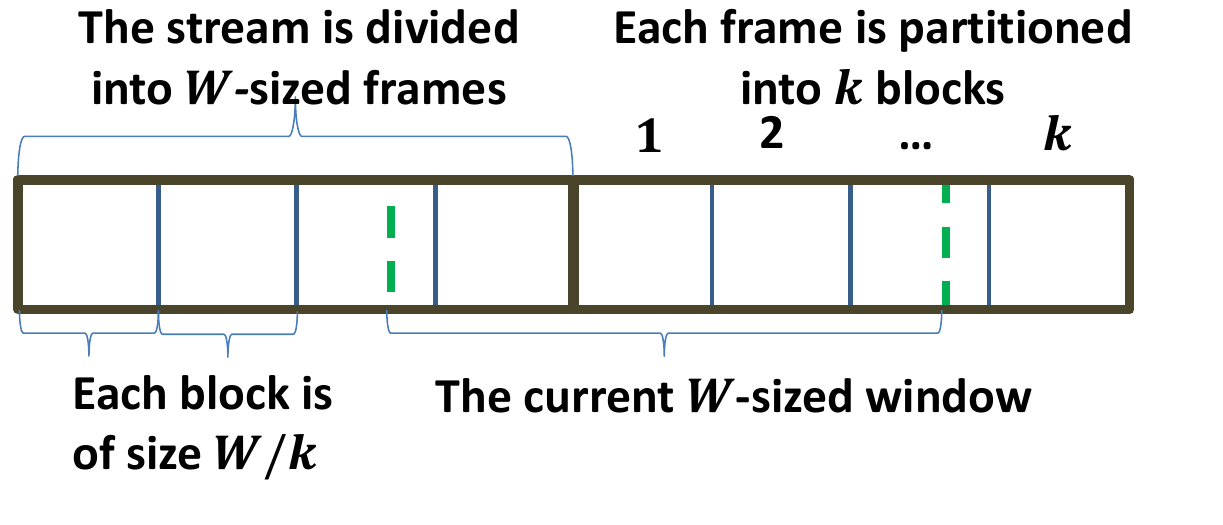}
    \caption{WCSS algorithm overview (adapting a figure in~\cite{ben2016heavy})}
    \label{fig:window_setup}
\end{figure}

\section{Learned Sliding Window Algorithm}
\label{sec:learned_wcss}

\subsection{Overview}
\label{sec:overview}

Given the specific characteristics of the sliding window setup, we want a predictor that offers a time-sensitive approach, in contrast to a predictor for whether an item is a heavy hitter over the stream, which is time-insensitive.
We employ an oracle that, for a specific item $i$ and timestamp $t$, predicts the next occurrence of $i$,
or equivalently how many arrivals will occur before item $i$ appears again.
A detailed description of this predictor, outlining its structure and functionality, will be provided later in this section.

Based on this prediction, \alg{} excludes items predicted to next appear again beyond the frame size from being inserted in the SS instance.
These items are referred to as \emph{non-essential items}.
As the SS instance is reset at the beginning of each frame, we obtain a ``cleaner'' SS instance, as proposed in~\cite{shahout2024learning}.
Note that the \emph{non-essential items} are not consistently ignored in every occurrence; rather we consider each arrival of the item and take an appropriate action based on the predictor. Note that if we had perfect predictions for non-essential items, then not including them in the Space Saving instance allows us to minimize memory overhead with no cost in terms of~accuracy.

However, machine learning methodologies are inherently imperfect, and they may exhibit errors, including substantial and unexpected errors. Therefore, as suggested in the algorithms with predictions literature (see, e.g.,~\cite{mitzenmacher2022algorithms}), algorithms based on machine learning predictions must demonstrate sufficient robustness to handle errors that may occur.  In particular, the notion of robustness that has become common 
is that the performance of algorithms using predictions should not be significantly inferior to that of conventional online algorithms that do not rely on predictions, even if predictions are inaccurate.
Unless some mitigating structure is added, ignoring items predicted as non-essential can lead to unbounded errors.
For example, if a heavy hitter of the current frame is predicted incorrectly as a non-essential item, this item is excluded from the sliding window algorithm, violating the error guarantee (Definition~\ref{definition:wfrequency_error_guarantee}).

We therefore apply the idea of keeping a Bloom filter of predicted non-essential items previously ignored during the current frame to ensure robustness, suggested by~\cite{shahout2024learning}.
The Bloom filter, like the Space Saving instance, is flushed at the beginning of each frame. 

There are two sources of estimation error when some arrivals of item $i$ are predicted to exceed the window size.
The first source of error arises when item arrivals are predicted incorrectly to exceed the window size. Ideally, these arrivals should have been captured by the Bloom filter and inserted back into the SS. However, when we insert item $i$ into the Bloom filter for the first arrival that exceeds the window size, we do not add it to the SS. Thus, if the first arrival is mispredicted, we may miss one insertion (otherwise, it is correct not to insert this arrival into the SS).
Second, the Bloom filter flushing at the beginning of a new frame causes an underestimation. Since the queried intervals can overlap with at most two frames, this imposes an underestimation error of no more than $1$. To compensate for these two sources of error occurring together, we add $2$ to the query.

Importantly, because the predictor is concerned with the item's next arrival, the number of \emph{single occurrence items} within the frame exceeds the number of single occurrence items across the entire stream and depends on the frame size as shown in Figure~\ref{fig:singles_ratio}.

\subsection{Next occurrence oracle}

\begin{figure}
    \centering
    \includegraphics[width = 5.5cm]{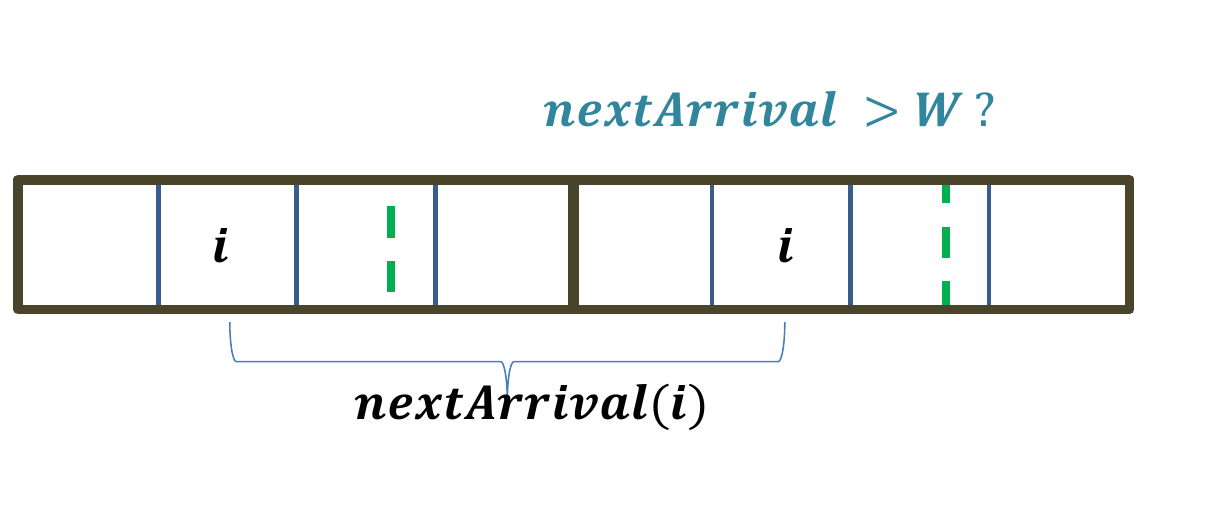}
    \caption{Next arrival prediction in sliding window setting.}
    \label{fig:learnedwindow_setup}
\end{figure}

In what follows, we use a slightly different approach that allows us to utilize a simpler predictor and still obtain strong performance.  Rather than try to predict the next arrival exactly, we predict whether the next arrival is larger than $W$ or not. 
If the next arrival is larger than $W$, then it clearly lies outside the frame (Figure~\ref{fig:learnedwindow_setup}).  This allows us to have our predictor perform a binary classification, rather than solve a regression problem.  


\begin{figure}
    \centering
    \includegraphics[width = 5.5cm]{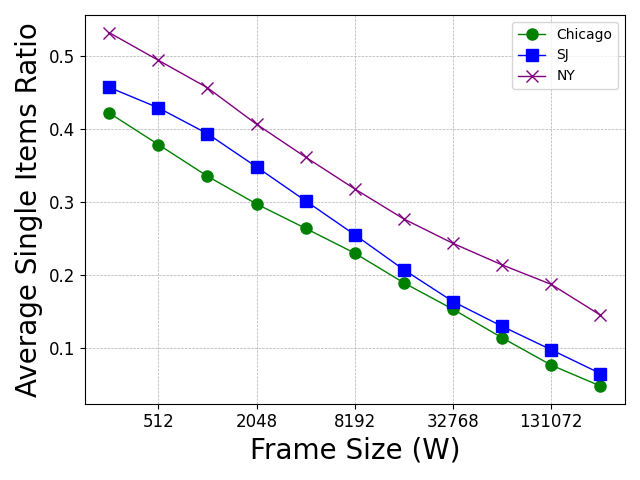}
    \caption{Average single items ratio vs. frame size using real-world traces (described in Section~\ref{sec:eval}).}
    \label{fig:singles_ratio}
\end{figure}

To construct the next occurrence oracle, we trained a neural network to predict if an item will show up again within a specified period of time (the next $W$ items in the stream). By converting the problem into a two-category classification problem, we simplify the prediction problem. The oracle is thus tasked to classify whether the next appearance of the item falls within our specified window size, tagged as 1, or exceeds this limit, indicated as 0.
Our model is a Recurrent Neural Network (RNN) that utilizes Long Short-Term Memory (LSTM) cell.  We chose this framework for its effectiveness in processing sequential data as demonstrated in prior research~\cite{kraska2018case}; however, we emphasize that other predictors could be used.

For our test cases, we focus on networking problems, and items are source-destination IP address pairs.  
The network begins by transforming the source and destination IP addresses into indices, which are then converted into embedding vectors. As a result of this process, IP addresses are represented in a dense and expressive manner. Our model uses these embeddings as input to the LSTM layer.
In the LSTM layer, the sequence of IP address embeddings is processed, in order to capture temporal dependencies. After the LSTM produces the final hidden states, the results are passed to a fully-connected (dense) layer that outputs one value.
This output value, after being passed through a sigmoid activation function, represents the model's prediction of whether the next appearance of an item falls within the specified window size $W$ or not. This prediction is generated based on the sequence of IP addresses provided as input.
The model is trained using Binary Cross-Entropy with Logits Loss (BCEWithLogitsLoss) and the Adam optimizer. After training, the model's performance is evaluated using the F1 score on a separate test dataset. 
Again, the implementation of such a predictor is our tailored approach, and it is just one of many possible options.
Depending on the requirements, our design may be replaced or enhanced with any other effective prediction~technique.

\subsection{Robustness Result}

Consider a sliding window algorithm $\mathcal A$ that follows the same setup of dividing the stream into frames. The WCSS algorithm is an example of such an algorithm. In the following, we prove the robustness of our approach in the general case using algorithm $\mathcal A$.
We show that our proposed algorithm is robust, in the sense that it cannot perform significantly worse than the corresponding algorithm that does not use predictions.

\begin{theorem}
\label{theorem:allow_correctness}
    Let $\mathcal A$ be an algorithm for $(W,\epsilon - \frac{2}{W})$-WFrequency. Then \alg{} solves \QWindowProblem{}.
\end{theorem}

\begin{proof}
For any item $i$ and any window of size $W$, we may underestimate the count of item $i$ by at most two.  This is because we may undercount an item once when it is predicted not to occur within the next $W$ steps and it is not already a positive from the Bloom filter.  As each window of length $W$ can intersect two frames, and the Bloom filter is reset to empty at frame boundaries, we can undercount an item at most twice over any $W$ consecutive steps.

We have \begin{align} \widehat{Q_i^{W}} \triangleq \mathcal{A}_{(W, \epsilon - \frac{2}{W})}.\mbox{{\sc Query$(i)$}} + 2.
\end{align}
That is, our estimate is the query result for $i$ from $\mathcal{A}$, which is an algorithm for
$(W,\epsilon - \frac{2}{W})$-WFrequency, with at most 2 instances of $i$ removed from the stream $\mathcal{A}$ processes and an extra count of 2 added back in.  It follows the smallest possible return value is 
$$(Q_i^{W}-2)+2 = Q_i^{W},$$ and the largest possible return value is 
$$\left (Q_i^{W} + \left (\epsilon - \frac{2}{W}\right )W \right ) + 2 = Q_i^{W} + \epsilon W,$$
proving the claim.
\end{proof}

\section{Evaluation}
\label{sec:eval}

In this section, we present an empirical study and compare WCSS and \alg{}.

\textbf{Experimental Setup.}
We implemented WCSS and LWCSS in Python 3.7.6.
The evaluation was performed on an AMD EPYC 7313 16-Core Processor with an NVIDIA A100 80GB PCIe GPU, running Ubuntu 20.04.6 LTS with Linux kernel 5.4.0-172-generic, and TensorFlow 2.4.1.
We extend WCSS by incorporating a predictor, which is built using an LSTM network, and compare it against traditional WCSS using real-world datasets.

\textbf{Datasets.} Our real datasets comprise CAIDA Internet Traces~\cite{CAIDACH16} obtained from multiple sources: (1) equinix-chicago 2016 high-speed monitor is located at an Equinix datacenter in Chicago and is connected to a backbone link of a Tier 1 ISP between Chicago, IL and Seattle, WA (2) equinix-sanjose 2014 monitor (SJ) is positioned at an Equinix datacenter in San Jose, CA, connected to a backbone link of a Tier 1 ISP between San Jose, CA and Los Angeles, CA (3) The equinix-nyc 2018 monitor is located at an Equinix datacenter in New York. These traces are summarized in Table~\ref{tbl:traces}.

For the synthetic data, we considered the Zipf distribution where the items are selected from a bounded universe, and the frequencies of an item with rank $R$ is given by $f(R, \alpha) = \frac{constant}{R^{\alpha}}$, where $\alpha$ is a skewness parameter.
For predictions, we employed a synthetic predictor that computes the true count for each item, and then adds an error which is drawn i.i.d. from a normal distribution with mean parameter 0 and standard deviation~$\sigma = 1$.

\begin{table}[t]
\centering
\caption{Traces from~\cite{CAIDACH16} used in the evaluation}
\begin{tabular}{l|l|l}
\toprule
Trace & \#Items & \#Uniques \\
\midrule
Chicago & 88,529,637 & 1,650,097 \\
 NY & 63,284,829 & 2,968,038  \\
 SJ & 188,511,031 & 2,922,904  \\
\bottomrule
\end{tabular}
\label{tbl:traces}
\end{table}

\begin{table}[t]
\centering
\caption{Predictor Accuracy}
\begin{tabular}{l|l}
\toprule
\textbf{Metric}          & \textbf{Value}               \\
\midrule
F1 Score (Chicago)       & 81.3\%                      \\
F1 Score (NY)       & 87.3\%                      \\
F1 Score (SJ)       & 83.5\%                      \\
\bottomrule
\end{tabular}
\label{tbl:predictor_eval}
\end{table}

\begin{figure*}[t]
\centering
	\begin{tabular}{cccc}
		\subfloat[Chicago]{\includegraphics[width = \smatrixCellWidth]{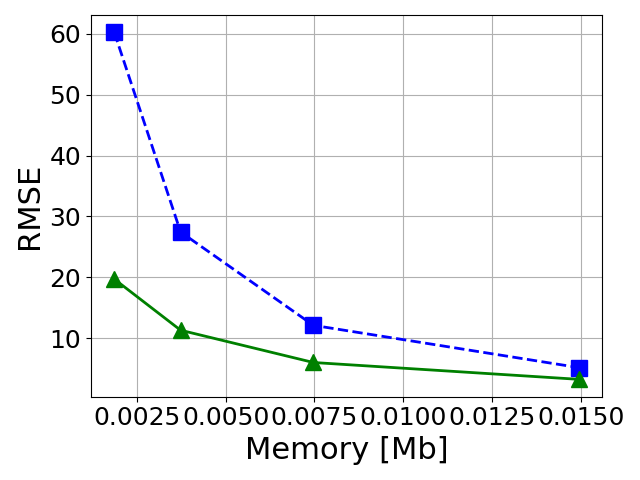}\label{fig:err_mem_chicago}} &
		\subfloat[NY]{\includegraphics[width = \smatrixCellWidth]{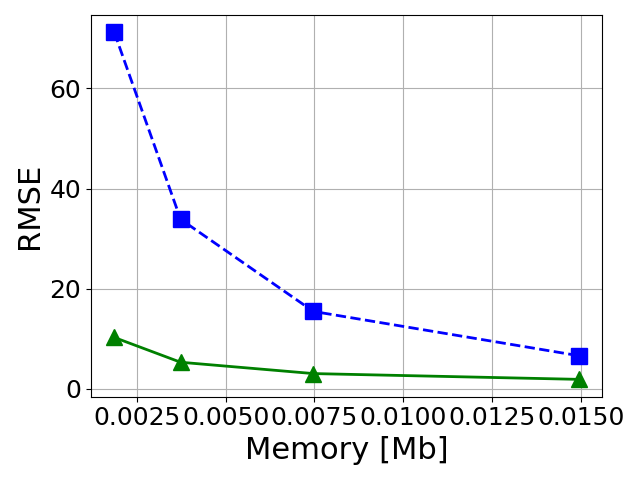}\label{fig:err_mem_ny}}
		\subfloat[SJ]{\includegraphics[width = \smatrixCellWidth]{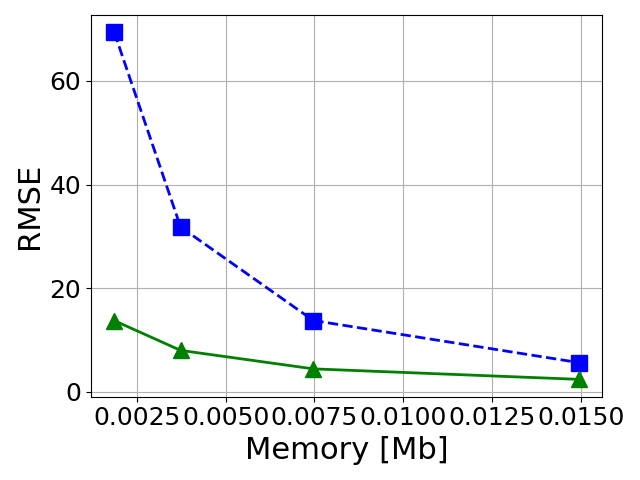}\label{fig:err_mem_sj}} &
          \subfloat[Synthetic]{\includegraphics[width = \smatrixCellWidth]{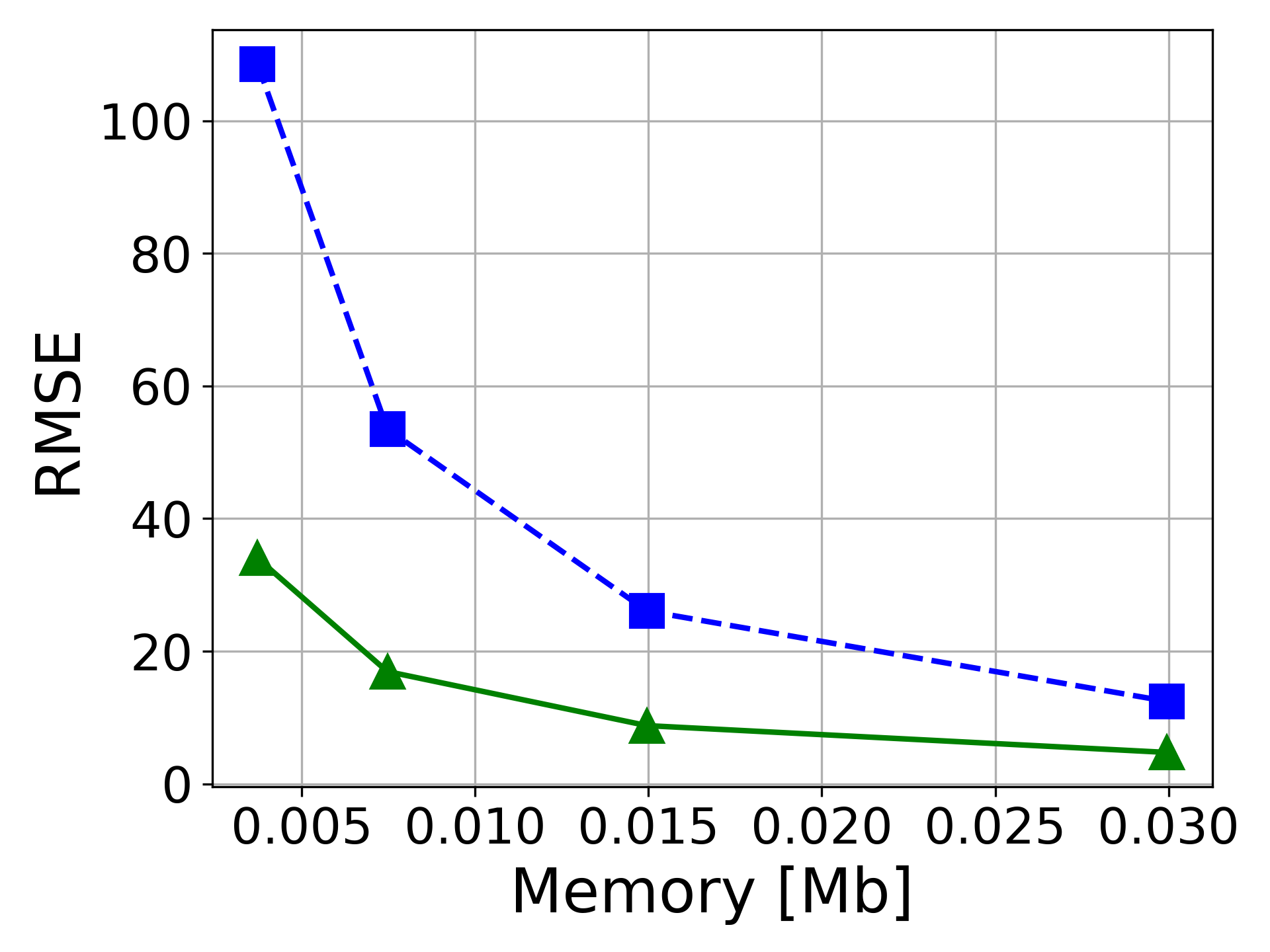}\label{fig:err_mem_synthetic}}
          \tabularnewline
        \multicolumn{3}{c}{\subfloat{\includegraphics[width = 3.5cm]{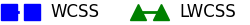}}}	
	\end{tabular}
	\caption{Accuracy comparison of WCSS and LWCSS vs. Memory (Megabytes) using real datasets (Chicago, NY and SJ)}
	\label{fig:lwcss_acc_vs_mem}
\end{figure*}

\textbf{Metrics.}
In the approximation error evaluation, we use Root Mean Square Error (RMSE). When an item is encountered, we query its frequency estimation immediately and calculate the RMSE.
For performance of operations, we use updates or queries per second.

\subsection{Accuracy Comparison}

Figure~\ref{fig:lwcss_acc_vs_mem} shows the accuracy (RMSE) as a function of memory for WCSS and LWCSS setting $W=2^{10}$. This evaluation includes three real datasets: Chicago, NY, and SJ, and utilizes a pre-trained LSTM predictor. The predictor's performance on each dataset is summarized in Table~\ref{tbl:predictor_eval}. The results demonstrate that LWCSS outperforms WCSS in terms of accuracy across all three datasets since the filtering of single-occurrence items enhances the SS's accuracy which is crucial for overall accuracy. Figure~\ref{fig:lwcss_acc_vs_mem} demonstratea the feasibility of our approach using the next arrival predictor.

\begin{figure*}[t]
	\center{
		\begin{tabular}{ccc}
                \subfloat{\includegraphics[width=\smatrixCellWidth]{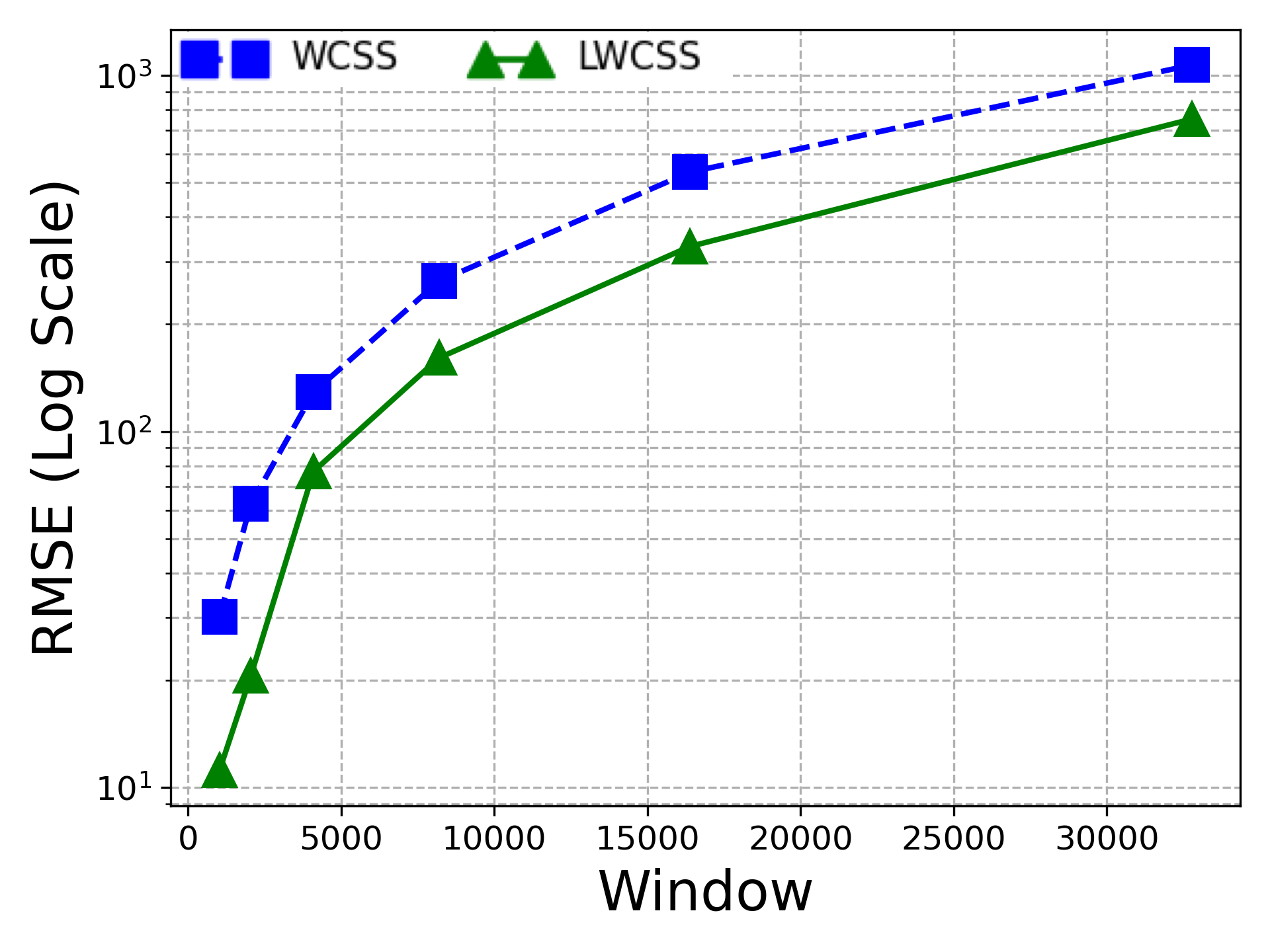} \label{fig:lwcss_acc_vs_window}}
               \subfloat[Query Performance]{\includegraphics[width=\smatrixCellWidth]{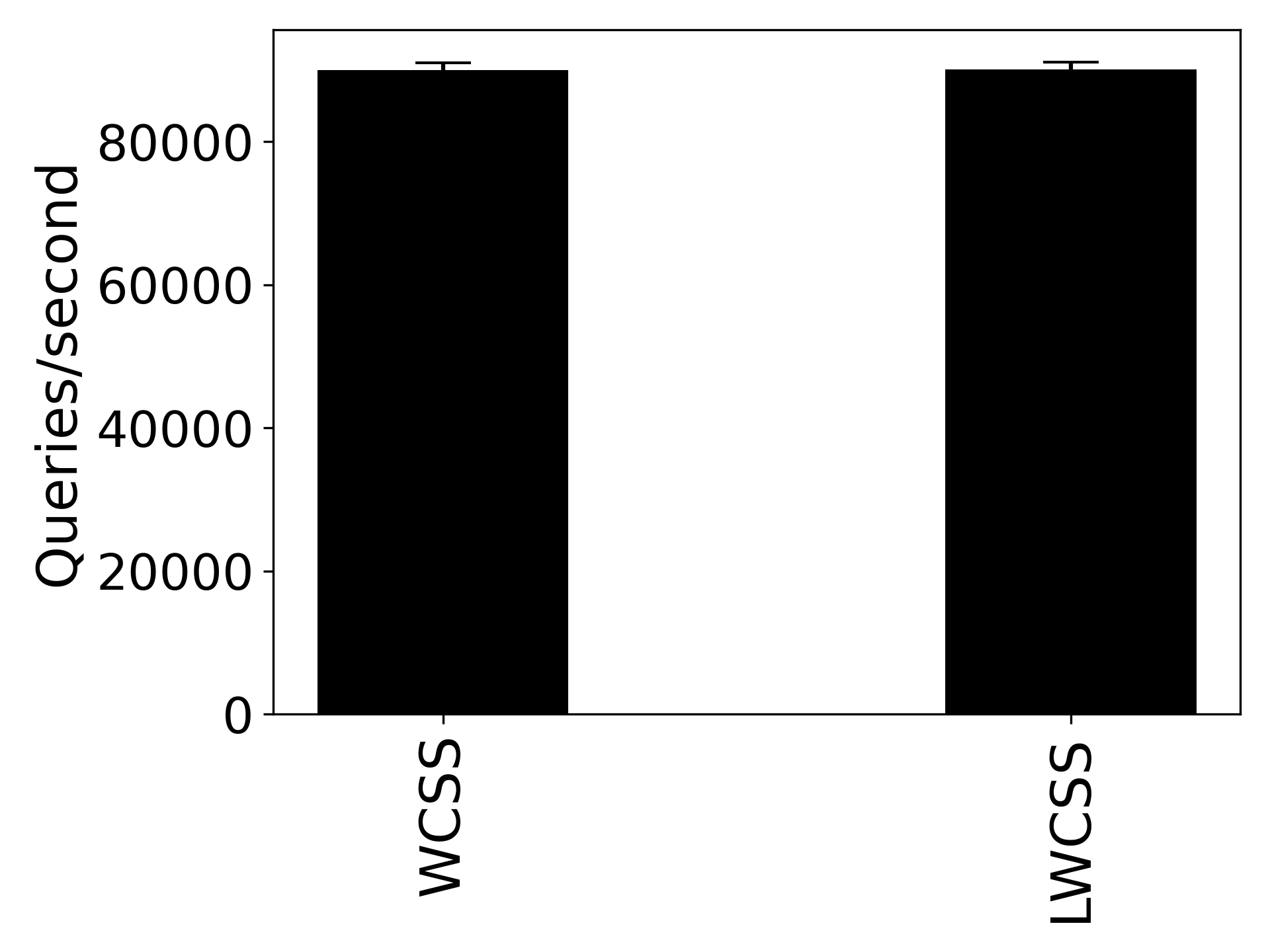}\label{fig:queries_perf}}
                \subfloat[Update Performance]{\includegraphics[width=\smatrixCellWidth]{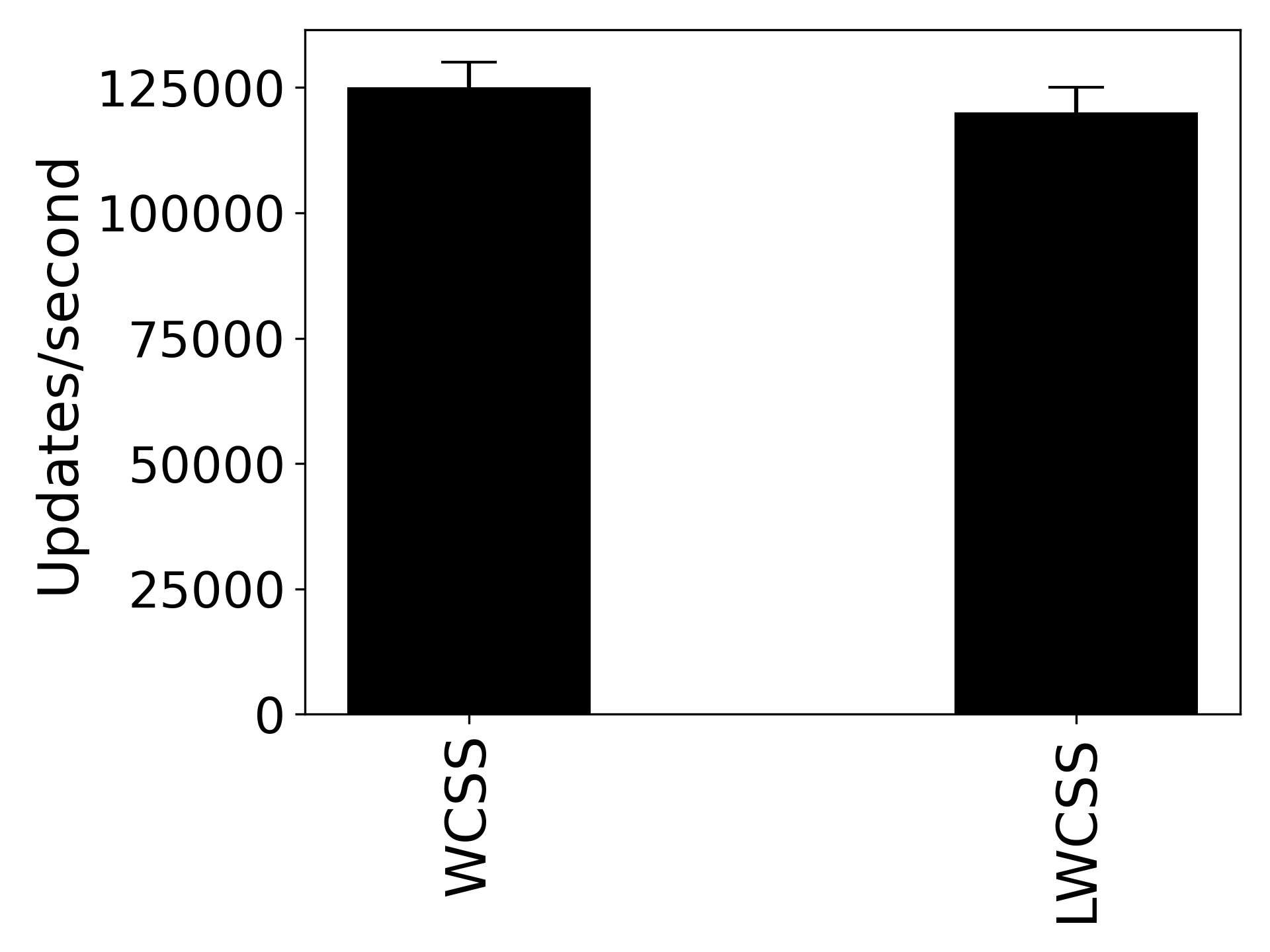}\label{fig:updates_perf}}
                \tabularnewline
		\end{tabular}
    		}
    \caption{Query and update performance using Chicago dataset and setting $W=2^{13}$.}
    \label{fig:controlled_params}
\end{figure*}

Figure~\ref{fig:lwcss_acc_vs_window} shows the RMSE as a function of window size using the Chicago dataset.
As seen, as the window size increases, the RMSE also increases. Additionally, the performance improves as the window is smaller because the ratio of single items becomes larger, as shown in Figure~\ref{fig:singles_ratio}.

\subsection{Performance Comparison}
Figures~\ref{fig:queries_perf} and~\ref{fig:updates_perf} show the performance of WCSS and LWCSS in terms of the update and query time respectively using the Chicago dataset and setting $W=2^{13}$.
The query performance of the two algorithms appears closely aligned due to their consistent operations. The update performance of \alg{} is slightly worse than WCSS due to the additional operation of inserting single occurence items into the Bloom filter.

\section{Future Work and Extensions}
\label{sec:future_works}


\subsection{An alternative approach: frequency prediction within time frames}

A potential direction for future work is to explore an alternative approach that learns item frequencies within each given time frame.
By predicting the frequencies of each item in the frame, we can apply previous work introduced by~\cite{hsu2019learning} to the sketch that tracks item frequencies within each frame. However, a challenge with this approach is that we need to retrain the predictor in each frame. One way to overcome this computational inefficiency in a streaming setting is: 
Instead of training a new predictor for each frame, we could employ a single predictor and retrain it periodically to adapt to shifts in the data distribution across frames.

Specifically, we can perform the retraining process after observing items for $M$ frames. The value of $M$ will be determined based on the extent of accuracy drop in the predictor; smaller values of $M$ should be set if the predictor incurs a higher accuracy drop. This periodic retraining aims to strike a balance between maintaining prediction accuracy and minimizing computational overhead.
To enhance the efficiency of the retraining process, we can utilize recent advances in transfer learning within deep learning~\cite{Ben12}. In particular, we can avoid retraining the model from scratch and shorten the training cycle by selectively retraining only specific layers in the model using information from recent frames. The main intuition behind this approach is that in many neural networks, the initial layers capture general features, while the later layers focus on specific features that are more dependent on the problem at hand (in our case, these later layers are sensitive to distribution shifts).
During the retraining process, we can freeze all hidden layers in the LSTM and fully-connected networks and then selectively retrain only those layers connected to the input and output of each network. As long as the dimensions of the frozen layers remain constant at any point in time, it is acceptable to reuse them. By preserving the learned general features and selectively updating the task-specific layers, we can leverage the knowledge gained from previous frames while efficiently adapting to the current data distribution.



As a future direction, we will compare this frame-based frequency prediction approach with our existing method using the next arrival predictor. This analysis could provide insights into the trade-offs between the two approaches and guide the selection of the appropriate technique based on the specific window size.

\subsection{General learned sliding window framework}
Stream processing encompasses three fundamental tasks: membership queries, frequency queries, and heavy hitter queries. Membership queries determine whether a given item is present within the sliding window. Frequency queries report the occurrence count of a specific item. Heavy hitter queries identify all items whose frequencies exceed a predetermined threshold.
Our research focuses on frequency queries over data streams within a sliding window setting. We recognized that the work proposed by~\cite{gou2020sliding} introduced a generic framework, termed Sliding Sketches, which can be applied to existing solutions for the above three tasks, enabling them to support queries over sliding windows. This work employed a similar approach of dividing the stream into frames (referred to as segments in their work). Consequently, our approach could intuitively be applied to this general framework.

\section{Related Work}
\label{sec:related}

The problem of estimating item frequencies over sliding windows was first studied in~\cite{ArasuM04}.
To estimate frequency within a $W\epsilon$ additive error over a window of size $W$, their algorithm requires $O(\oneOverE\log^2\oneOverE\log W)$ bits of memory. This memory requirement was later optimized to $O(\oneOverE\log W)$ bits as highlighted in~\cite{LeeT06}. Hung and Ting in~\cite{HungLT10} further refined this by improving the update time to a constant and locating all heavy hitters in the optimal $O(\oneOverE)$ time. The WCSS algorithm, as introduced in~\cite{ben2016heavy}, also provides frequency estimates in constant time.

Algorithms with predictions is, as we have stated, a rapidly growing area.  The site \cite{awpweb} contains a collection of over a hundred papers on the topic.  
The idea of using predictions to specifically improve frequency estimation algorithms appears to have originated with \cite{hsu2019learning}, where they augmented a learning oracle of the heavy hitters into frequency estimation sketch-based algorithms. Later~\cite{jiang2019learning} and ~\cite{shahout2024learning} explored the power of such an oracle, showing that it can be applied to a wide array of problems in data streams.
All these papers use neural networks to learn certain properties of the input. We, however, differ from those papers because we consider the sliding window setup, in which properties derived from window resolution can differ significantly from those derived from the entire stream, and therefore, other predictions are required.

\section{Conclusion}
\label{sec:conclusion}

We have presented a novel method to improve sliding window algorithms for approximate frequency estimation by incorporating a learning model that filters out low-frequency ``noisy'' items. 
Past research on learning-augmented algorithms does not perform well in the sliding window settings due to variations between the properties of the sliding window resolution and the entire data stream.

We have demonstrated the benefits of our design both analytically and empirically.

\noindent\textbf{Code Availability:} All code is available online~\cite{opensource}.

\section*{Acknowledgments}
Rana Shahout was supported in part by Schmidt Futures Initiative and Zuckerman Institute. Michael Mitzenmacher was supported in part by NSF grants CCF-2101140, CNS-2107078, and DMS-2023528.
\newpage

\bibliographystyle{plain}
\bibliography{refs}

\end{document}